\documentclass[11pt]{article}

\usepackage{amsmath,amsthm,amssymb,amsfonts,verbatim}
\usepackage{latexsym}
\usepackage{fullpage}
\usepackage{graphicx}
\usepackage{float}
\usepackage{enumerate}

\usepackage{thmtools}
\usepackage{thm-restate}
\usepackage[dvipsnames,usenames]{color}
\usepackage[pdftex,letterpaper=true,colorlinks=true,pdfpagemode=none,urlcolor=blue,linkcolor=blue,citecolor=BrickRed,pdfstartview=FitH]{hyperref}
\usepackage{hyperref}
\usepackage{times}
\usepackage{cite}
\usepackage{cleveref}

 \providecommand{\F}{\mathbb{F}}

\date{}

\newtheorem{lemma}{Lemma}

\newtheorem{prop}[lemma]{Proposition}

\newtheorem{claim}[lemma]{Claim}
\newtheorem{defn}{Definition}

\newtheorem{rmk}{Remark}

\def \mC {\mathcal{C}}

\def \mC {\mathcal{C}}

\def \Xi {{X^{[i]}}}

\newcommand{\Z}{\mathbb{Z}}

\begin{document}

\onecolumn

\title{Explicit Constructions of Two-Dimensional Reed-Solomon Codes in High Insertion and Deletion Noise Regime}

\author{Tai Do Duc\thanks{Tai Do Duc is with the Division of Mathematical Sciences, School of Physical and Mathematical Sciences, Nanyang Technological University, Singapore 637371 (email: doductai001@e.ntu.edu.sg)}, Shu Liu\thanks{Shu Liu is with the National Key Laboratory of Science and Technology on Communications, University of Electronic Science and Technology of China, Chengdu 611731, China (email: shuliu@uestc.edu.cn).}, 
Ivan Tjuawinata\thanks{Ivan Tjuawinata is with the Strategic Centre for Research on Privacy-Preserving Technologies and Systems, Nanyang Technological University, Singapore 637553 (ivan.tjuawinata@ntu.edu.sg).}
~and Chaoping Xing
\thanks{ Chaoping Xing is with the Division of Mathematical Sciences, School of Physical and Mathematical Sciences, Nanyang Technological University,  Singapore 637371 (email: xingcp@ntu.edu.sg).}}
\maketitle

\begin{abstract}
Insertion and deletion (insdel for short) errors are synchronization errors in communication systems caused by the loss of positional information in the message. 
Reed-Solomon codes have gained a lot of interest due to its encoding simplicity, well structuredness and list-decoding capability~\cite{GS1999} in the classical setting. This interest also translates to the insdel metric setting, as the Guruswami-Sudan decoding algorithm~\cite{GS1999} can be utilized to provide a deletion correcting algorithm in the insdel metric~\cite{SW2003}. 
%On top of its wide interest in its error correcting capability in the classical setting, Reed-Solomon codes also prove to have error correcting algorithms under insdel metric which is based on Guruswami-Sudan decoding algorithm~\cite{GS1999}. 
Nevertheless, there have been few studies on the insdel error-correcting capability of Reed-Solomon codes.

Our main contributions in this paper are explicit constructions of two families of $2$-dimensional Reed-Solomon codes with insdel error-correcting capabilities asymptotically reaching those provided by the Singleton bound. The first construction gives a family of Reed-Solomon codes with insdel error-correcting capability asymptotic to its length. The second construction provides a family of Reed-Solomon codes with an exact insdel error-correcting capability up to its length. Both our constructions improve the previously known construction of $2$-dimensional Reed-Solomon codes \cite{TS2007} whose insdel error-correcting capability is only logarithmic on the code length. 

\end{abstract}

\section{Introduction}~\label{introduction}

Insertion and deletion (insdel for short) errors are synchronization errors~\cite{stoc2017},~\cite{HSS}  in communication systems caused by the loss of positional information of the message. Without loss of generality, we can consider only insertions and deletions for synchronization errors, since a substitution can be replaced by a deletion followed by an insertion and this at most doubles the number of operations. Insdel codes have recently attracted lots of attention due to their applicabilities in many interesting fields such as DNA storage, DNA analysis~\cite{SJFFH},  \cite{RW2005}, race-track memory error correction \cite{CK2017} and language processing \cite{BM2000}, \cite{Och03}. 

In the insdel setting, the distance between two codewords is the smallest number of insertions and deletions needed to transform one codeword into the other codeword. The minimum insdel distance of a code is the minimum insdel distance among its codewords. The study of codes with insdel errors was pioneered by Levenshtein, Varshamov and Tenengolts in the 1960s~\cite{VT65},\cite{1965},\cite{1967} and \cite{1984}. Sloane~\cite{NJA} constructed a family of codes capable of correcting single deletion. He focused on binary codes and discussed difficulties of constructing deletion correcting codes. By using combinatorial designs, Bours~\cite{Bours} constructed codes of lengths $4$ and $5$ which correct up to $2$ and $3$ deletions, respectively. More constructions of codes of length $5$ capable of correcting $3$ deletions are given by Mahmoodi~\cite{mah}. Using incomplete directed designs, Yin~\cite{Yin} and Shalaby et al.~\cite{Shalaby} gave constructions of codes of length $6$ capable of correcting $4$ deletions. Recently, there have been several constructions of codes with high insdel error-correcting capabilities (or insdel list decoding capabilities) through synchronization strings~\cite{stoc2017},\cite{HS2018}, \cite{HSS} and concatenation~\cite{GW11},\cite{VR16},\cite{Japan},\cite{2019insdel}. 

%{\color{red} Despite the presence of an efficient deletion correcting algorithm for (generalized) Reed-Solomon codes in the insdel metric \cite{GS1999, SW2003}, there have been few constructions of Reed-Solomon codes with high insdel error-correcting capabilities. 
%There has been much less improvements on Reed-Solomon codes with good properties on insdel metric. It is observed \cite{SW2003,2004} that those families of Reed-Solomon codes  \cite{S1997,GS1999,RR2000,KV2003,GV2005} studied in Hamming metric also have deletion correcting capabilities in the insdel metric.}
In addition to extensive studies on the decoding of (generalized) Reed-Solomon codes \cite{S1997},\cite{GS1999},\cite{RR2000},\cite{KV2003},\cite{GV2005}, it has also been observed \cite{SW2003},\cite{2004} that these codes have an efficient decoding algorithm against deletion errors. Despite the existence of such algorithm, construction of Reed-Solomon codes to rectify insdel errors has relatively been less explored.
 In this paper, we aim to provide a more thorough investigation on Reed-Solomon codes in the insdel metric.

%{\color{red} (this sentence is repeted with previous results.)} 

%{\color{red}An important advantage of using generalized Reed-Solomon ccodes,initially noted in, is the existence of an efficient deletion correcting algorithm. The decoding algorithm for generalized Reed-Solomon codes can be formulated as a polynomial reconstruction problem, to which the efficient list decoding algorithm of Guruswami and Sudan applies.}

\subsection{Previous results}
The Singleton bound \cite{stoc2017} implies that a Reed-Solomon code of length $n$ and dimension $k$ has deletion error-correcting capability at most $n-k$.
There have been few constructions of Reed-Solomon codes with high capabilities, but none has come close to this bound. In 2004, Y. Wang, L. McAven and R. Safavi-Naini~\cite{2004} constructed a class of Reed-Solomon codes of length $5$ and dimension $2$ which can correct one deletion.
They also provided numerical results on deletion correcting capabilities of generalized Reed-Solomon codes.
Their construction relies on the rank analysis of matrices where the number of matrices increases exponentially on the code length. Hence it may not be feasible to generalize this approach to a larger length.
At the same year, D. Tonien and R. Safavi-Naini~\cite{TS2007} gave constructions of deletion correcting codes by using Reed-Solomon codes and their subcodes. 
%{\color{red}Their constructions are only explicit given the existence of a certain polynomial used to generate the evaluation points of the codes. }
%{\color{red}A Reed-Solomon code constructed in this method is only explicit given the existence of a certain polynomial used to generate its evaluation points.
%In general, the existence of such a polynomial is unknown. This construction provides a Reed-Solomon code of length $n$ and dimension $k$ with deletion error-correcting capability $\log_{k+1} n -1$ which grows much more slowly than $n-k$, the error-correcting capability provided by the Singleton bound. Furthermore, even when subcodes of Reed-Solomon codes are considered, the deletion error-correcting capability is at most $\frac{n}{k+1}-1$.}
% {\color{magenta} The family of Reed-Solomon codes constructed in this paper has deletion error-correcting capability $\log_{k+1}n-1$, where $k$ and $n$ are dimension and length of the code, respectively. Furthermore, subcodes of these Reed-Solomon codes are also studied and there are subcodes with deletion error-correcting capability n/(k+1)-1. We note that both of the constructed Reed-Solomon codes and subcodes have deletion error-correcting capabilities much smaller than $n-k$, the capability provided by the Singleton bound. Furthermore, these codes are only explicit ... is unknown.}
They constructed a family of Reed-Solomon codes of length $n$ and dimension $k$ with deletion error-correcting capability $\log_{k+1}n-1.$ They also provided studies on subcodes of Reed-Solomon codes and obtained subcodes with deletion error-correcting capability $\frac{n}{k+1}-1.$ Note that all the constructed codes have deletion error-correcting capabilities much smaller than $n-k,$ the Singleton bound on deletion error-correcting capability of an $[n,k]$ Reed-Solomon code. Moreover, these constructions are only explicit given the existence of a certain polynomial used to generate their evaluation points. In general, the existence of such a polynomial is unknown. 

In 2007, L. McAven and R. Safavi-Naini~\cite{MS2007} provided an upper bound for the insdel error-correcting capability of Reed-Solomon codes with dimension $2$ over prime fields. They show that if the code has length $n\geq 3$, then its insdel error-correcting capability is at most $n-3$. This bound is tighter than the Singleton bound, which states that the insdel error-correcting capability of the same code is at most $n-2$. As a result, Reed-Solomon codes of length $n\geq 3$ and dimension $2$ over prime fields can never achieve the Singleton bound.

\subsection{Our results}
%{\color{magenta} Let $\mC$ be a code and let $u$ and $v$ be any two words in $\mC$. The insdel distance between $u$ and $v$ is the minimum number of insertions and deletions which is needed to transform $u$ into $v$. The insdel distance of $\mC$ is the minimum insdel distance between any two distinct codewords in $C$. }{\color{green} Need?}

The Singleton bound \cite{stoc2017} implies that any $[n,k]$ linear code over $\F_q$ has insdel distance $d\leq 2n-2k+2$. Our first result in this paper generalizes the result by L. McAven and R. Safavi-Naini~\cite{MS2007} mentioned above. We prove that if $q$ is \textit{large enough} compared to $n$, then a $k$-dimensional Reed-Solomon code of length $n$ over $\F_q$ can never achieve the equality $d=2n-2k+2$ in the Singleton bound. 

\begin{restatable}{thm}{improvedbound}\label{improvedbound}
Let $\mC$ be a $k$-dimensional Reed-Solomon code of length $n$ over $\F_q$ with $2\leq k <n$. Let $d$ denote the minimum insdel distance of $\mC$. If $q\geq n^2$, then
\begin{equation} \label{insdeldistance}
d\leq 2n-2k.
\end{equation}
As a result, these codes can only correct at most $n-k-1$ insdel errors. 
\end{restatable}

%Though Theorem~\ref{improvedbound} implies that Reed-Solomon codes defined over a finite field $\F_q$ whose size is large enough compared to the length $n$ can never achieve the Singleton bound, it has no implication that Reed-Solomon codes cannot attain insdel error-correcting capabilities close to those obtained the Singleton bound.

Theorem~\ref{improvedbound} implies that for a large enough alphabet, an insdel Reed-Solomon code can never achieve the Singleton bound. A natural question is whether the upper bound of insdel distance given in Theorem~\ref{improvedbound} can be achieved by some Reed-Solomon codes.
The answer of this question is affirmative in the case of $2$-dimensional Reed-Solomon codes.
We construct two classes of $2$-dimensional Reed-Solomon codes whose insdel error-correcting capabilities are close to $n-3$, the capability provided by Theorem \ref{improvedbound}.
In the first construction, our code can correct up to $n(1+o(1))$ insdel errors. In the second construction, our code can correct up to exactly $n-3$ insdel errors. While the first construction is more self-contained and does not require any prior knowledge in number theory, the second one uses ideas in the study of cyclotomic numbers \cite{do2019}.

\begin{restatable}{thm}{mainone} \label{main1}
Let $\epsilon>0.$
When $q$ is sufficiently large, there exists an explicit family of $2$-dimensional Reed-Solomon codes over $\F_q$ with length $n$ and minimum insdel distance $d$ satisfying
$$n=\exp(\log(q)^\epsilon) \ \ \text{and} \ \ d\geq 2(1-\epsilon)n+O(1/\epsilon^2).$$ 
\end{restatable}

%{\color{blue} Observe that the construction in Theorem~\ref{main1} only achieves the limit of insdel error-correcting capability asymptotically and it requires $q$ to be super-polynomial in $n.$ The second construction {\color{red}in Theorem~\ref{main2}} provides codes that achieves the limit of insdel error-correcting capability even for small $n.$ However, it provides a restriction for the characteristic of $\F_q$ which is only satisfied when the characteristic is at least $5.$}

\begin{restatable}{thm}{maintwo} \label{main2}
Let $q=p^s$ be a power of a prime $p.$ If $q-1$ contains a nontrivial divisor $f$ which satisfies
\begin{itemize}
\item[(i)] $f<\log_{\sqrt{14}}\left( p^{{\rm{ord}}_{f}(p)}\right)$, and
\item[(ii)] $f$ is proportional to $\log(q)$,
\end{itemize}
then there exists a $2$-dimensional Reed-Solomon code $\mC$ over $\F_q$ with length $n$ and insdel distance $d$ satisfying
$$n=\Theta\left( \sqrt{\log q} \right) \ \ \text{and} \ \ d=2n-4.$$
\end{restatable}

%{\color{red} Theorem 1,2 and 3 are repeat to Theorem 4, 5 and 6. I think could we use three paragraphs separately describe those three theorems results. If Theorem 1, 2 and 3 changed, then the organization should be changed.}

Note that the construction in Theorem~\ref{main1} provides a family of $2$-dimensional Reed-Solomon codes whose minimum insdel distance is asymptotically best possible, i.e. $d=2n(1+o(1))$. For this construction to work, we require $q$ to be super-polynomial in $n$. On the other hand, the construction in Theorem~\ref{main2} provides a family of $2$-dimensional Reed-Solomon codes whose minimum insdel distance is exact and meets the bound provided in Theorem~\ref{improvedbound}, i.e., $d=2n-4$. For this construction to work, we require $q-1$ to have a nontrivial divisor $f$ which is logarithmic on the characteristic $p$ of $\F_q$.

\subsection{Organization}
This paper is organized as follows. In Section~\ref{sec:prelim}, we introduce some definitions and basic results on insdel codes and Reed-Solomon codes.
Moreover, we discuss results on cyclotomic numbers and Singer difference sets which will be needed for our construction in Section \ref{sec:constr2}.
In Section \ref{sec:imp}, we prove an improved version of the Singleton bound on the insdel error-correcting capabilities of Reed-Solomon codes and provide a general framework for our constructions in Section \ref{sec:constr1} and Section \ref{sec:constr2}.
Lastly, Section \ref{sec:constr1} is used to prove Theorem \ref{main1} and Section \ref{sec:constr2} is used to prove Theorem \ref{main2}.

\section{Preliminaries}\label{sec:prelim}

\subsection{Insdel codes and Reed-Solomon codes}
Let $\F_q$ be a finite field with $q$ elements and $\F_q^n$ be the set of all vectors of length $n$ over $\F_q.$ For any positive real number $i,$ we denote by $[i]$ the set of positive integers $\{1,\cdots,\lfloor i\rfloor\}.$ 

\begin{defn}[Insdel distance]
The insdel distance $d(\mathbf{a},\mathbf{b})$ between two words $\mathbf{a}\in \F_q^{n_1}$ and $\mathbf{b}\in \F_q^{n_2}$ (not necessarily of the same length) is the minimum number of insertions and deletions which is needed to transform $\mathbf{a}$ into $\mathbf{b}.$ It can be verified that $d(\mathbf{a}, \mathbf{b})$ is a metric.
%The insdel distance of a code $\mC$ is the minimum distance between any two codewords in $\mC$.
\end{defn}

\begin{defn}
An insdel code $\mathcal{C}$ of length $n$ is a subset of $\F_q^n.$ Its minimum insdel distance is defined as 
\begin{eqnarray*}
\displaystyle d(\mC)=\min_{\mathbf{c_1},\mathbf{c_2}\in \mathcal{C},\mathbf{c_1}\neq \mathbf{c_2}}\{d(\mathbf{c_1},\mathbf{c_2})\}.
\end{eqnarray*}
\end{defn}

\noindent Note that for two distinct codewords $\mathbf{c}_1,\mathbf{c}_2$ in an insdel code $\mathcal{C}$ of length $n$, the distance $d(\mathbf{c}_1,\mathbf{c}_2)$ is an even number between $2$ and $2n.$
Throughout this paper, given two vectors $\mathbf{a}$ and $\mathbf{b}$ and a code $\mathcal{C},$ unless specifically noted, we use $d(\mathbf{a},\mathbf{b})$ and $d(\mathcal{C})$ to denote the insdel distance between $\mathbf{a}$ and $\mathbf{b}$ and the minimum insdel distance of the code $\mathcal{C}$, respectively.

\medskip

It is well known that the insdel distance between two vectors can be calculated via their longest common subsequence. We state this result in the following lemma and include a proof for the convenience of the reader.

\begin{lemma} \label{ins}
If $\mathbf{a}\in \F_q^{n_1}$ and $\mathbf{b} \in \F_q^{n_2}$, where $n_1$ and $n_2$ are not necessarily the same, we have 
\begin{equation} \label{insd}
d(\mathbf{a},\mathbf{b})= n_1+n_2-2\ell,
\end{equation}
where $\ell$ denotes the length of a longest common subsequence of $\mathbf{a}$ and $\mathbf{b}$.
\end{lemma}

\begin{proof}
Let ${\mathbf{c}}$ be a longest common subsequence of $\mathbf{a}$ and $\mathbf{b}.$ We can transform $\mathbf{a}$ to $\mathbf{b}$ by first transforming $\mathbf{a}$ to ${\mathbf{c}}$ via $n_1-\ell$ deletions, then $\mathbf{b}$ to $\mathbf{c}$ via $n_2-\ell$ insertions. So $d(\mathbf{a}, \mathbf{b})\leq d(\mathbf{a}, \mathbf{c})+d(\mathbf{b},\mathbf{c})\leq n_1+n_2-2\ell $. 
So, to prove Equation (\ref{insd}), it suffices to show that $d(\mathbf{a},\mathbf{b})\geq n_1+n_2-2\ell$. 

Suppose that $d(\mathbf{a},\mathbf{b})=d_D+d_I$, that is, we have a way to transform $\mathbf{a}$ to $\mathbf{b}$ via $d_D$ deletions and $d_I$ insertions. Note that up to positional adjustments, insertion and deletion operations are commutable. Hence, we can assume that all deletions are done first before the insertions are conducted. Supposing $\mathbf{w}$ is the resulting subword after the $d_D$ deletions from $\mathbf{a},$ it is easy to see that ${\mathbf{w}}$ is a common subsequence of $\mathbf{a}$ and $\mathbf{b}$ and it has length $\ell'=n_1-d_D=n_2-d_I.$ By the maximality of $\ell,$ we have that $\ell'=n_1-d_D=n_2-d_I\leq \ell.$ Thus, $d(\mathbf{a},\mathbf{b})= d_I+d_D\geq n_1+n_2-2\ell.$ We finish the proof of Lemma \ref{ins}.
\end{proof}

An insdel code over $\F_q$ of length $n$, size $M$ and minimum insdel distance $d$ is called an $(n,M,d)_q$-insdel code. Moreover, we call a code an $[n,k,d]_q$-insdel code if it is a $k$-dimensional linear code over $\F_q$ of length $n$ and minimum insdel distance $d$. 

\medskip

The minimum insdel distance of an insdel code is one of its most important parameters, as it indicates the insdel error-correcting capability of the code.
%It is desirable that $d$ is as large as possible with an insdel code with fixed length $n.$ 
For a fixed length $n$ and a fixed size $M,$ it is desirable for an $(n,M,d)_q$- insdel code to have $d$ as large as possible.
It is shown in \cite{stoc2017} that an $(n,M,d)_q$-insdel code $\mathcal{C}$ must obey the following version of the Singleton bound.

\begin{prop}[Singleton Bound~\cite{stoc2017}]
Let $\mathcal{C}\subseteq \F_q^n$ be an $(n,M,d)_q$-insdel code of length $n$ and minimum insdel distance $0\leq d\leq 2n.$ Then
\begin{equation*}
M\leq q^{n-d/2+1.}
\end{equation*}
As a consequence, if $\mC$ is an $[n,k,d]_q$-insdel code, then
\begin{equation} \label{singleton-insdel}
d\leq 2n-2k+2.
\end{equation}
\end{prop}

The minimum insdel distance $d$ of a code provides its insdel error-correcting capability. In the following we will define the error-correcting capability of a code. 

%\begin{defn}\label{def:dcc}
%A code $\mathcal{C}\subseteq \F_q^n$ is called a $t$-deletion correcting code if for any $\mathbf{v}\in\F_q^{n-t},$ it is a subsequence of at most one codeword $\mathbf{c}\in\mathcal{C}.$ In other words, transmitting any codeword $\mathbf{c}\in\mathcal{C}$ through a channel which can incur up to $t$-deletion, we will always be able to recover the correct transmitted codeword.
%\end{defn}

%The definition can be extended to include both deletion and insertion error.

\begin{defn}\label{def:idcc}
A code $\mathcal{C}\subseteq \F_q^n$ is called a $t$-insdel error-correcting code if for any word $\mathbf{v}\in \bigcup_{m=n-t}^{n+t}\F_q^m,$ there is at most one codeword $\mathbf{c}\in\mathcal{C}$ such that $\mathbf{v}$ can be obtained from $\mathbf{c}$ by performing $t_I$ insertions and $t_D$ deletions to $\mathbf{c}$ for some non-negative integers $t_I$ and $t_D$ with $t_I+t_D\leq t.$ In other words, transmitting any codeword $\mathbf{c}\in\mathcal{C}$ through a channel which can incur up to $t$ insdel errors, we will always be able to recover the correct transmitted codeword. A $t$-deletion error-correcting code can also be similarly defined.
\end{defn}

%{A $t$-deletion correcting code can also be similarly defined.}
%Although Definition~\ref{def:idcc} seems to give a stronger restriction on the code than the one in Definition~\ref{def:dcc}, it can actually be shown that they are equivalent.
The following remark provides the equivalence between $t$-deletion error-correcting code and $t$-insdel error-correcting code.
\begin{rmk}[\cite{Bours},\cite{mah}]
A code $\mathcal{C}\subseteq\F_q^n$ is a $t$-deletion error-correcting code if and only if it is a $t$-insdel error-correcting code. Due to this equivalence, we use deletion error-correcting capability and insdel error-correcting capability interchangeably in our discussion.
% {\color{blue} Note that although the deletion error correcting and insdel correcting capability of a code is the same, its deletion correcting algorithm is not the same as its insdel correcting algorithm.}{\color{green} Do we need to stress this again here to show that although deletion correcting algorithm exists, it does not directly imply the insdel correcting algorithm?} {\color{red} For our paper, we have not focused on the decoding algorithm so why we need to mention algorithm here.}
\end{rmk}

The relation between insdel error-correcting capability and minimum insdel distance of {an insdel} code is the same as that of a classical code. %{\color{red} Assume that through $t$ steps of insertions and deletions, a word $v\in \F_q^{m}$ can be transformed into two distinct codewords $c_1$ and $c_2$ of $\mC$. We have $d(c_1,c_2)\leq 2t$, as we can transform $c_1$ into $v$ after $t$ steps and $v$ into $c_2$ after another $t$ steps. The Inequation $d(c_1,c_2)\leq 2t$ cannot happen if $t\leq \lfloor (d-1)/2\rfloor$, where $d$ denotes the minimum insdel distance of the code $\mC$. We have the following.} 
Assume that a word $\mathbf{r}$ has insdel distance at most $t$ to both codewords $\mathbf{c}_1$ and $\mathbf{c}_2$ of $\mathcal{C}.$ If $d(\mathcal{C})=d,$ the triangle's inequality implies that $d\leq d(\mathbf{c}_1,\mathbf{c}_2)\leq 2t.$ {So, to ensure uniqueness of the codeword which is within $t$ insdel operations from any given received word, we set $t$ to be an integer less than $\frac{d}{2}.$} This {discussion} is summarised in the following remark.

\begin{rmk}
A code $\mathcal{C}\subseteq \F_q^n$ of minimum insdel distance $d$ has insdel error-correcting capability up to $\left\lfloor \frac{d-1}{2}\right\rfloor.$
\end{rmk}

Next, we provide the definition of Reed-Solomon {codes}. 

\begin{defn}
Let $q$ be a prime power and let $n,k$ be positive integers such that $k\leq n\leq q$. Let $S=\{\alpha_1,\cdots, \alpha_n\}$ be a subset of $\F_q$ of size $n$. We define $\F_q^{< k}[x]$ to be the set of polynomials over $\F_q$ of degree at most $k-1.$  For any $f(x)\in \F_q[x]$, define $\mathbf{c}_{f,S}=(f(\alpha_1),\cdots, f(\alpha_n))\in \F_q^n$. The Reed-Solomon code $RS_{n,k,S}$ is defined to be the set of vectors $\mathbf{c}_{f,S}$ with $f(x)\in \F_q^{< k}[x],$ that is,
\begin{equation*}
RS_{n,k,S}= \{(f(\alpha_1),\cdots, f(\alpha_n)): f(x)\in\F_q^{<k}[x]\}.
\end{equation*}
\end{defn}

\noindent Note that under Hamming metric, the minimum Hamming distance of $RS_{n,k,S},$ denoted by $d_{\sf H}(RS_{n,k,S})$, always achieve{s} the Singleton bound, that is, $d_{\sf H}(RS_{n,k,S})= n-k+1$.
Howeover, this is not necessarily true for the case of insdel distance, i.e. the insdel distance $d$ of $RS_{n,k,S}$ may be strictly smaller than $2n-2k+2$.
Theorem \ref{improvedbound} confirms this assertion when $q$ is large enough compared to $n$. Prior to our result, McAven and Safavi-Naini \cite{MS2007} provided an evidence on this in the case $k=2$, $n\geq 3$ and $q$ is a prime.

\medskip

\subsection{Cyclotomic numbers and Singer difference sets}

In this subsection, we prove a refined version of \cite[Main Theorem 1]{do2019}  and state a result on Singer difference sets which will be  needed for our construction in Section \ref{sec:constr2}. First, we state some notations and definitions.

Let $q$ be a power of a prime $p$. Let $e$ and $f$ be nontrivial divisors of $q-1$ such that $q=ef+1$. By $\zeta_f$ we denote a primitive $f$th root of unity. By $\Z_f$ we denote the set of residues modulo $f$, i.e. $\Z_f=\{0,1,\dots,f-1\}$.
Let $g$ be a primitive element of the finite field $\F_q$. For each $a\in \Z$, write
$S_a=\{g^a,g^{a+e},\dots,g^{a+(f-1)e}\}$.
For $a,b\in \Z$, define $(a,b)$ as a cyclotomic number of order $e$ whose value is equal to the number of solutions to the following equations
$$1+x=y, \ x\in S_a, \ y\in S_b.$$
Equivalently, we have
\begin{eqnarray*}
(a,b)&=&|(1+S_a)\cap S_b| \\
&=& |\{(i,j)\in \Z_f\times \Z_f: 1+g^{a+ie}=g^{b+je}\}|.
\end{eqnarray*}
Note that $S_a=S_{a+e}$ for any integer $a$. Hence, we only need to consider the sets $S_a$ and cyclotomic numbers $(a,b)$ with $a,b\in \{0,1,\dots,e-1\}$.
By \cite[Main Theorem 1]{do2019}, we have $(a,b)\leq 3$ if $p>\left(\sqrt{14}\right)^{f/{\rm{ord}}_f(p)}$.
The idea for the proof of this result is to transform equations over $\F_q$ into equations over $\mathbb{C}$ and then study the later equations. The following version of \cite[Theorem 4.1]{do2019} is needed for our purpose.

\begin{lemma}\cite[Theorem 4.1]{do2019} \label{equiv}
Assume that the polynomial ${h}(x)=\sum_{i=0}^{f-1}a_ix^i \in \Z[x]$ has coefficients $a_i$ satisfying $\sum_{i=0}^{f-1} a_i^2\geq 3$. If
$$p^{{\rm{ord}}_{f}(p)}>\left(\sum_{i=0}^{f-1}a_i^2\right)^{f/2},$$
then 
$${h}(g^e)=0 \ \text{over} \ \F_q \Leftrightarrow {h}(\zeta_f)=0 \ \text{over} \ \mathbb{C}.$$
\end{lemma}

\medskip

The study of cyclotomic numbers $(a,b)$ is divided into five cases: $(0,0)$, $(0,a)$, $(a,0)$, $(a,a)$ and $(a,b)$, where $a$ and $b$ are distinct elements in the set $\{1,\dots,e-1\}$. The following lemma summarizes the results of these cases, see \cite[Theorems 5.1, 5.3, 5.4, 5.5, 5.6]{do2019}.

\begin{lemma}\cite{do2019} \label{cyclo_known}
Let $a$ and $b$ be distinct integers in the set $\{1,\dots, {e}-1\}$. If $f<\log_{\sqrt{14}}\left(p^{{\rm{ord}}_f(p)}\right)$, then we have the following
\begin{eqnarray*}
(0,0) &\leq & \begin{cases} 3 \ \mathrm{if} \ f\equiv 0\pmod{6} \ \mathrm{and} \ 2\in S_0,\\ 2 \ \mathrm{otherwise}. \end{cases} \\
(0,a) &\leq & \begin{cases} 3 \ \mathrm{if} \ 2\in S_a, \\ 2 \ \mathrm{if} \ 2\not\in S_a. \end{cases} \\
(a,0) &\leq & \begin{cases} 3 \ \mathrm{if}\ 2\in S_{a} \ \mathrm{and} \ 2\mid f, \\ 2 \ \mathrm{otherwise}.\end{cases}\\
(a,a) &\leq & \begin{cases} 3 \ \mathrm{if} \ 2\in S_{-a} \ \mathrm{and} \ 2\mid f, \\ 2 \ \mathrm{otherwise}.\end{cases}\\
(a,b) &\leq & 2.
\end{eqnarray*}
\end{lemma}

\medskip

We will prove a result stronger than \cite[Main Theorem 1]{do2019} as follows. 

\begin{lemma} \label{cyclo_refined}
Let $a$ and $b$ be any two elements in $\{0,1,\dots,e-1\}$. If $f<\log_{\sqrt{14}}\left(p^{{\rm{ord}}_{f}(p)}\right)$, then 
$$(a,b)\leq 2.$$
\end{lemma}
\begin{proof}
By Lemma \ref{cyclo_known}, it suffices to prove that $(0,0)\leq 2$, $(0,a)\leq 2$, $(a,0)\leq 2$ and $(a,a)\leq 2$, where $a \in \{1,\dots,{e}-1\}$. 

First, assume that $(0,0)\geq 3$. By Lemma \ref{cyclo_known}, we obtain $(0,0)=3$, $f\equiv 0\pmod{6}$ and $2\in S_0$. Let $t$ be the unique integer in the set $\{0,\dots, f-1\}$ such that $2=g^{te}\in S_0$. Note that 
$$(0,0) = |\{(i,j)\in \Z_f\times \Z_f: 1+g^{ie}=g^{je}\}|.$$
As $(0,0)=3$, there exist $i,j\in \Z_f$ such that $i\neq 0$, $j\neq t$ and $1+g^{ie}=g^{je}$. As $g^{te}=2$, we obtain $1+g^{ie}=2g^{(j-t)e}$, which implies
\begin{equation}\label{overFq}
2-g^{(t-j)e}-g^{(t-j+i)e}=0.
\end{equation}
Note that the numbers $t-j$ and $t-j+i$ are calculated modulo $f$, as $g^{ef}=1$. Moreover, the numbers $0$, $t-j$ and $t-j+i$ are pairwise different by our choice of $i$ and $j$. Let ${h}(x)=\sum_{i=0}^{f-1}a_ix^i \in \Z[x]$ be a polynomial which satisfies ${h}(g^e)=2-g^{(t-j)e}-g^{(t-j+i)e}$. The polynomial ${h}(x)$ has one coefficient equal to $2$, two coefficients equal to $-1$ and all remaining coefficients equal to $0$. Thus $\sum_{i=0}^{f-1}a_i^2=6$. The condition $f<\log_{\sqrt{14}}\left( p^{{\rm{ord}}_f(p)}\right)$ implies $p^{{\rm{ord}}_f(p)}>\left(\sqrt{14}\right)^f> \left(\sum_{i=0}^{f-1}a_i^2\right)^{f/2}$. As ${h}(g^e)=0$ over $\F_q$ by {Equation} (\ref{overFq}), we obtain, by Lemma \ref{equiv}, 
$${h}(\zeta_f)=2-\zeta_f^{t-j}-\zeta_f^{t-j+i}=0.$$
Hence $\cos(2\pi (t-j)/f)+\cos(2\pi(t-j+i)/f)=2$, which happens only when $t-j=t-j+i=0$, contradicting the choices of $i$ and $j$. Thus $(0,0)\leq 2$.

\medskip

The proof of $(0,a)\leq 2$ is similar to that of $(0,0)\leq 2$. Assume that $(0,a)\geq 3$. By Lemma \ref{cyclo_known}, we obtain $2\in S_a$, that is, there exists an integer $t\in \{0,\dots, f-1\}$ such that $2=g^{a+te}\in S_a$.
As $(0,a)\geq 3$, there exist $i,j\in \Z_f$ such that $i\neq 0$, $j\neq t$ and $1+g^{ie}=g^{a+je}$. As $g^{a+te}=2$, we obtain $1+g^{ie}=2g^{(j-t)e}$, which implies
$$2-g^{(t-j)e}-g^{(t-j+i)e}=0.$$
The rest of the proof follows exactly the same as in the previous case.

\medskip

Next, assume that $(a,0)\geq 3$. By Lemma \ref{cyclo_known}, we have $(a,0)=3$, $2\in S_a$ and $f\equiv 0\pmod{2}$. Fix $t\in \{0,1,\dots, f-1\}$ such that $g^{a+te}=2$. In this case, we prove that $(a,0)=1$, which is a contradiction and confirms our claim that $(a,0)\leq 2$. Let $i,j\in \Z_f$ such that $1+g^{a+ie}=g^{je}$. Hence $1+2g^{(i-t)e}=g^{je}$, which implies $2+g^{(t-i)e}-g^{(t-i+j)e}=0$. By similar reasoning as in the previous case, we obtain
$$2+\zeta_f^{t-i}-\zeta_f^{t-i+j}=0,$$
which implies $2+\cos\left(2\pi (t-i)/f\right)-\cos (2\pi (t-i+j)/f)=0$.
Hence $i=t+f/2$ and $j=f/2$. Therefore, the only solution to the equation $1+g^{a+xe}=g^{ye}$, $x,y\in \Z_f$, is $x=t+f/2$ and $y=f/2$, contradicting the assumption $(a,0)\geq 3$. We obtain $(a,0)\leq 2$.

\medskip

Lastly, we prove $(a,a)\leq 2$. For each pair $(i,j)\in \Z_f\times\Z_f$ with $1+g^{a+ie}=g^{a+je}$, we have $1+g^{-a-ie}=g^{(j-i)e}$. Hence, each solution $(i,j)\in \Z_f\times \Z_f$ to the equation $1+g^{a+xe}=g^{a+ye}$, $x,y\in\Z_f$, induces a solution $(-i,j-i)$ to the equation $1+g^{-a+xe}=g^{ye}$. We obtain $(a,a)=(-a,0)$. The inequality $(a,a)\leq 2$ follows from the case of $(-a,0)$.
\end{proof}

\medskip

The last result in this {sub}section is by Singer \cite{sin1938}. 
%{It will be used for our construction in Section \ref{sec:constr2} {\color{blue} to provide a large code length $n$}.} 
It will serve as one of the main components in our construction in Section~\ref{sec:constr2}.

\begin{lemma}[{\cite{sin1938} or \cite[Theorem 1.2.10]{sch2002}}] \label{singer}
Let $r$ be a prime power and let $d\geq 3$ be an integer. Let $\rm{Tr}_{\F_{r^d}/\F_r}$ denote the trace function from $\F_{r^d}$ to $\F_r$. Then the set
$$D=\{x\F_r^*: x\in \F_{r^d}^*, \rm{Tr}_{\F_{r^d}/\F_r}(x)=0\}$$
is a subset of $G=\F_{r^d}^*/\F_r^*$ with the properties that $|D|=(r^{d-1}-1)/(r-1)$ and each non-identity element of $G$ appears exactly once in the multi-set $DD^{(-1)}=\{xy^{-1}: x,y\in D\}.$
\end{lemma}

\noindent We remark that the set $D$ defined in Lemma \ref{singer} is called a Singer difference set.

\bigskip

\section{An improved upper bound and framework for our constructions} \label{sec:imp}

In this section, we prove Theorem \ref{improvedbound} and provide a general framework for our constructions in Section \ref{sec:constr1} and Section \ref{sec:constr2}.

\subsection{An improved upper bound {on minimum insdel distance}}
First, we recall Theorem \ref{improvedbound}.

%{\color{red}\begin{theorem} 
%Let $\mC$ be a $k$-dimensional Reed-Solomon code of length $n$ over $\F_q$ with $2\leq k <n$. Let $d$ denote the minimum insdel distance of $\mC$. If $q\geq n^2$, then
%\begin{equation}
%d\leq 2n-2k.
%\end{equation}
%As a result, these codes can only correct at most $n-k-1$ insdel errors. 
%\end{theorem}}

\improvedbound*

%\subsection*{Proof of Theorem \ref{improvedbound}}

\begin{proof}
Let $\alpha_1,\dots,\alpha_n\in \F_q$ be $n$ distinct elements which define the code $\mC$, that is,
$$\mC=\{\left(f(\alpha_1),\dots,f(\alpha_n)\right): f\in \F_q[X]^{<k}\}.$$
Let $\mC'$ be the subcode of $\mC$ such that each codeword in $\mC'$ has pairwise distinct entries. Hence, each codeword in $\mC'$ has the form
$\left(\sum_{i=0}^{k-1}c_i\alpha_1^i,\dots, \sum_{i=0}^{k-1}c_i\alpha_n^i\right)$ such that
\begin{itemize}
\item[(i)] $c_i \in \F_q$ for $0\leq i\leq k-1$,
\item[(ii)] $c_1\not\in \Big\{-\sum_{i=2}^{k-1}c_i\frac{\alpha_j^i-\alpha_t^i}{\alpha_j-\alpha_t}: \ \alpha_j\neq \alpha_t\Big\}$ if $k\geq 3$, and $c_1\neq 0$ if $k=2$.
\end{itemize}
The second condition implies that there are at least $q-\frac{n(n-1)}{2}$ choices for $c_1$ for any given choice of $c_0,c_2,\cdots,c_{k-1}.$
Thus,
\begin{equation} \label{size}
|\mC'| \geq q^{k-1}\left(q-\frac{n(n-1)}{2}\right)> q^{k-1}\left( q-\frac{n^2}{2}\right).
\end{equation}
Put $\ell=\frac{d}{2}-1.$ Any two subwords of length $n-\ell$ of a codeword $\mathbf{u}\in \mC'$ are distinct, as $\mathbf{u}$ has all positions distinct.
Moreover, any two subwords of length $n-\ell$ of two different codewords $\mathbf{u},\mathbf{v}\in \mC'$ are also distinct because otherwise, we would have by Equation \eqref{insd}, $d(\mathbf{u},\mathbf{v})\leq 2n-2(n-\ell)=2\ell=d-2<d$, contradiction.
Thus, we can form more than $q^{k-1}(q-n^2/2){n\choose n-\ell}$ ordered words of length $n-\ell$, each having all positions distinct, from the code $\mC'$.
As the number of ordered words (over $\F_q$) of length $n-\ell$ with all positions distinct is $q(q-1)\cdots(q-(n-l)+1)$, we obtain
\begin{equation}\label{inequal}
q^{k-1}(q-n^2/2){n\choose n-\ell} < q(q-1)\dots (q-(n-\ell)+1)\leq q^{n-\ell}.
\end{equation}
Note that $\ell=d/2-1\leq n-1$, as $d\leq 2n$. Furthermore, we can assume $\ell\geq 1$. If $\ell=0$, then $d=2\leq 2n-2k$. When $\ell\geq 1,$ we have $n\leq {n\choose n-\ell}$ and $q-n^2/2\geq q/2$. By \eqref{inequal}, we obtain
$$\frac{nq^k}{2}<q^{n-\ell},$$
which implies $n-\ell \geq {k+1}$, that is, $\ell=d/2-1\leq n-k-1$.
\end{proof}

\subsection{Framework for our constructions}
Our goal in the remainder of this paper is to construct $2$-dimensional insdel Reed-Solomon codes of length $n$ and minimum insdel distance {approaching} $2n-2$, the distance provided by the Singleton bound.
%whose insdel distance is close to that of Singleton bound $2n-4$, where $n$ is length of the code. 
We show that this problem can be reduced to finding a subset $S$ of $\F_q$ which has small intersections with all its nontrivial linear transformations,
%translations
that is $|\delta S\cap (S+\gamma)|$ is small for any $(\delta, \gamma) \in \F_q\times \F_q \setminus \{(1,0)\}.$

\medskip

Consider a Reed-Solomon code $\mC$ of dimension $2$ under insdel metric as 
$$\mC=\{(a\alpha_1+b, a\alpha_2+b,\cdots, a\alpha_n+b): a,b\in\F_q\},$$ 
where $S=\{\alpha_1, \alpha_2, \cdots, \alpha_n\}$ is a subset of $\F_q$ of size $n.$
%Let {\color{red}$\mathbf{u}=(a\alpha_1+b,\dots,a\alpha_n+b)$ and $\mathbf{v}=(c\alpha_1+d,\dots,c\alpha_n+d)$ be any two distinct codewords in $\mC$, that is, $(a,b)\neq (c,d)$ and $a,b,c,d\in \F_q$.
%By \eqref{insd}, the insdel distance between $\mathbf{u}$ and $\mathbf{v}$ is $d(\mathbf{u},\mathbf{v})=2n-2\ell$, where $\ell$ is the length of the longest common subword of $\mathbf{u}$ and $\mathbf{v}$.
%In other words, there exist $1\leq i_1<\dots<i_\ell\leq n$ and $1\leq j_1<\dots<j_\ell\leq n$ such that
%\begin{equation} \label{commonsubword}
%a\alpha_{i_t}+b=c\alpha_{j_t}+d \ \mathrm{for} \ t=1,\dots,\ell.
%\end{equation}
%If $a=0$ or $c=0$, it is clear that $l\leq 1$. Hence $d(u,v)=2n-2l \geq 2n-2$. From now on, assume that $a\neq 0$ and $c\neq 0$. Put $\delta=a/c$ and $\gamma=(d-b)/c$. Note that $(\delta, \gamma)\neq (1,0)$, as $(a,b)\neq (c,d)$. The equation \eqref{commonsubword} implies
%$$\delta\alpha_{i_t}=\alpha_{j_t}+\gamma \ \text{for} \ t=1,\dots,\ell.$$
%We obtain $\ell\leq |\delta S\cap (S+\gamma)|$, which implies that the insdel distance between $u$ and $v$ is lower bounded by $2n- 2|\delta S\cap (S+\gamma)|$, i.e.
%\begin{equation} \label{insdellowerbound}
%d(u,v) \geq 2n- 2|\delta S\cap (S+\gamma)|.
%\end{equation}
%From now on, we focus on finding a subset $S$ of $\F_q$ of size $n$ such that $|\delta S\cap (S+\gamma)|$ is small for any $(\delta, \gamma)\in \F_q\times \F_q\setminus \{(1,0)\}$.
%} {\color{green} change of notation to avoid notation confusion where $c$ can be entry of codeword and $d$ can be minimum insdel distance of a code.}
Let ${\bf c}_1=(a_1\alpha_1+b_1, a_1\alpha_2+b_1,\cdots, a_1\alpha_n+b_1)$ and ${\bf c}_2=(a_2\alpha_1+b_2, a_2\alpha_2+b_2,\cdots, a_2\alpha_n+b_2)$ be {any} two distinct codewords of $\mC$, that is, $a_1, a_2, b_1, b_2\in \F_q$ such that $(a_1,a_2)\neq(b_1,b_2).$ By Equation~\eqref{insd}, the insdel distance between $\mathbf{c}_1$ and $\mathbf{c}_2$ equals to $2(n-\ell)$ where $\ell$ is the length of {a} longest common subsequence of $\mathbf{c}_1$ and $\mathbf{c}_2.$ In other words, there exist $1\leq i_1<\dots<i_\ell\leq n$ and $1\leq j_1<\dots<j_\ell\leq n$ such that
\begin{equation} \label{commonsubword}
a_1\alpha_{i_t}+b_1=a_2\alpha_{j_t}+b_2 \ \mathrm{for} \ t=1,\dots,\ell.
\end{equation}
If {either $a_1=0$ or $a_2=0$}, it is clear that $\ell\leq 1$ which implies $d(\mathbf{c}_1,\mathbf{c}_2)= 2n-2\ell\geq 2n-2.$ The remaining case is when $a_1$ and $a_2$ are both non-zero. Put $\delta=a_1/a_2$ and $\gamma=(b_2-b_1)/a_2.$ Note that $(\delta,\gamma)\neq (1,0)$ since $(a_1,b_1)\neq (a_2,b_2).$ Equation~(\ref{commonsubword}) implies 
\begin{equation*}
\delta \alpha_{i_t}=\alpha_{j_t}+\gamma~\mathrm{~for~}t=1,\cdots,\ell.
\end{equation*}
So $\ell\leq |\delta S\cap (S+\gamma)|,$ which implies
\begin{equation}\label{insdellowerbound}
d(\mathbf{c}_1,\mathbf{c}_2)\geq 2n-2|\delta S\cap (S+\gamma)|.
\end{equation}
{Hence}, the construction of insdel Reed-Solomon codes with high minimum insdel distance {is} reduced to the construction of a subset $S$ of $\F_q$ of size $n$ {such that $|\delta S\cap (S+\gamma)|$ is small} for any $(\delta,\gamma)\in \F_q\times \F_q\setminus\{(1,0)\}.$

\section{Reed-Solomon Codes with Insdel Error-Correcting Capability Asymptotic to its Length} \label{sec:constr1}

In this section, we prove Theorem~\ref{main1}, which 
%{\color{red} claims that there exist a family of Reed-Solomon codes of dimension $2$  whose insdel distances are asymptotic to those provided by the Singleton bound  We restate this result in the following theorem.}
provides an explicit construction of $2$-dimensional Reed-Solomon codes with minimum insdel distance asymptotically close to the Singleton bound.
%{\color{blue}\mainone*}
%{\color{red}\begin{theorem} \label{main11}
%Let $\epsilon>0.$
%Then when $q$ is sufficiently large, there exists an explicit family of $2$-dimensional Reed-Solomon codes over $\F_q$ with length $n$ and minimum insdel distance $d$ satisfying
 %$$n=\exp(\log(q)^\epsilon) \ \ \text{and} \ \ d\geq 2(1-\epsilon)n+O(1/\epsilon^2).$$
%\end{theorem}}
The first step in our construction is to construct an explicit family of Reed-Solomon codes with insdel distance up to its length $n$. This construction is then generalized to obtain Reed-Solomon codes with insdel distance asymptotic to $2n$.

\begin{lemma}  \label{EC}
Let $\F_q$ be a finite field with $q$ elements of characteristic $p$ such that $q=p^e$ for some positive integer $e>1$. Furthermore, assume that $e$ can be written as a product of two factors that are coprime and proportional to each other with the smallest being at least $5$. Then there exists an explicit family of $2$-dimensional Reed-Solomon codes over $\F_q$ of length $n=\exp(\sqrt{\log q}),$ with minimum insdel distance of at least $n-2.$
\end{lemma}

\begin{proof}

The Inequality \eqref{insdellowerbound} implies that for the insdel distance of $\mC$ to be at least $n-2$, it suffices to find a subset $S$ of $\F_q$ of size $n$ such that 
\begin{equation} \label{lowerbound}
|\delta S\cap (S+\gamma)|\leq \frac{n}{2}+1 \ \text{for any}  \ (\delta,\gamma)\in \F_q\times \F_q\setminus \{(1,0)\}.
\end{equation}
Suppose that $e=t_1t_2$, where $5\leq t_1<t_2$, $\gcd(t_1,t_2)=1$ and $t_1,t_2$ are proportional to each other. Let $F_1= \F_{p^{t_1}}$ and $F_2 = \F_{p^{t_2}}.$ Suppose that for $i=1,2, F_i=\mathrm{span}_{\F_p}\{1,\beta_i,\cdots,\beta_i^{t_i-1}\}.$ Define
 $$G_i^\ast=\mathrm{span}_{\F_p}\{\beta_i,\cdots,\beta_i^{t_1-1}\}\setminus\{0\}, \ G_i=\beta_i+\mathrm{span}_{\F_p}\{\beta_i^2,\cdots,\beta_i^{t_1-1}\}\subseteq G_i^\ast$$ 
 and 
 \begin{equation} \label{setS}
S=G_1\cup G_2. 
 \end{equation}
 Note that $F_1\cap F_2=\F_p$, $|G_1\cap G_2|=0$ and $|G_1|=|G_2|=p^{t_1-2}$.
  Hence, $n=|S|=|G_1\cup G_2|=2p^{t_1-2}$. We prove that the set $S$ defined in Equation \eqref{setS} satisfies Inequality \eqref{lowerbound}.
 
   Fix $\delta,\gamma\in\F_q$ such that $(\delta,\gamma)\neq (1,0).$ We can assume that $\delta\neq 0$, as otherwise we have $|\delta S\cap (S+\gamma)|\leq 1$, which satisfies Inequality \eqref{lowerbound}. For $i,j\in\{1,2\},$ define $\Theta_{i,j}=\delta G_i \cap (G_j+\gamma).$ Then $\delta S\cap(S+\gamma) = \bigcup_{i,j\in\{1,2\}} \Theta_{i,j}$ and 
   \begin{equation} \label{union}
   |\delta S\cap(S+\gamma)| \leq |\Theta_{1,1}|+|\Theta_{1,2}|+|\Theta_{2,1}|+|\Theta_{2,2}|.    
   \end{equation}
Before proving the bounds on $|\Theta_{i,j}|$, we give the following observations which will be used repeatedly in our analysis later.

\begin{itemize}
\item[(i)] If there exist $a,b\in F_i$ such that $\delta a = b+\gamma,$ we have that $\delta \in F_i$ if and only if $\gamma\in F_i.$ So, if $|\Theta_{i,i}|>0,$ we must have either both $\delta$ and $\gamma$ to be in $F_i$ or both not in $F_i.$
\item[(ii)] The set $\{1,\beta_i,\cdots,\beta_i^{t_1-1}\}$ is $\F_p$-linearly independent since they form a subset of a basis of $\F_{p^{t_i}}$ over $\F_p.$ Hence, $\F_p\cap G_1^*=\F_p\cap G_2^*=\emptyset$ and $G_1^*\cap F_2=F_1\cap G_2^*=G_1^*\cap G_2^*=\emptyset.$
\item[(iii)] For any $\lambda\in \F_p^*$ and $a\in G_i^\ast$, we have $\lambda a\in G_i^\ast$. On the other hand, if both $a$ and $\lambda a$ belong to $G_i,$ then $\lambda=1.$
\end{itemize}
The proofs of (i)-(iii) are straightforward and omitted. The following claim gives upper bounds on the numbers $|\Theta_{1,2}|$ and $|\Theta_{2,1}|$.

\begin{claim} \label{lemmaneq}
$|\Theta_{1,2}|\leq p$ and $|\Theta_{2,1}|\leq p$.
\end{claim}

\noindent We prove Claim \ref{lemmaneq} by contradiction. Suppose that $|\Theta_{1,2}|\geq p+1.$ Let $a_1,a_2,x_1,\dots,x_{p-1}$ be $p+1$ distinct elements of $G_1$ and $b_1,b_2,$ $y_1,\dots,y_{p-1}$ be $p+1$ distinct elements of $G_2$ such that $\delta a_i=b_i+\gamma$ for $i=1,2$ and $\delta x_j=y_j+\gamma$ for $j=1,\dots,p-1$.
Now for any $j=1,\dots,p-1,$ we have 
$$\delta=\frac{b_1-b_2}{a_1-a_2}=\frac{b_1-y_j}{a_1-x_j}=\frac{b_2-y_j}{a_2-x_j}.$$ 
Note that $a/b=c/d$ if and only if there exists $\lambda$ such that $a=\lambda c$ and $b=\lambda d.$ So, there exists $\lambda_j,\mu_j\in \F_q\setminus\{0\}$ such that $b_1-b_2=\lambda_j(b_1-y_j), b_1-b_2=\mu_j(b_2-y_j), a_1-a_2=\lambda_j(a_1-x_j)$ and $a_1-a_2=\mu_j(a_2-x_j).$
Furthermore, note that $\lambda_j=(b_1-b_2)/(b_1-y_j)\in F_2\setminus\{0,1\}$ and $\lambda_j=(a_1-a_2)/(a_1-x_j)\in F_1\setminus\{0,1\}.$ Since $F_1\cap F_2=\F_p,$ $\lambda_j\in \F_p\setminus\{0,1\}$. Similarly, we have $\mu_j\in \F_p\setminus\{0,-1\}$.
The equations $b_1-b_2=\lambda_j(b_1-y_j)$ and $b_1-b_2=\mu_j(b_2-y_j)$ yield the following system of linear equations:

\begin{equation} \label{sle}
\left.
\begin{array}{rcrcrcl}
(1-\lambda_j)b_1 &-& b_2 & + & \lambda_j y_j &=&0\\
b_1 & - &(1+\mu_j) b_2 & + & \mu_jy_j &=&0
\end{array}
\right\}.
\end{equation}

Consider two vectors $\mathbf{u}=(1-\lambda_j,-1,\lambda_j)$ and $\mathbf{v}=(1,-(1+\mu_j),\mu_j).$ Note that they are a multiple of each other if and only if $\lambda_j+\lambda_j \mu_j = \mu_j.$ If they are not a multiple of each other, then the system in Equation~\eqref{sle} implies that $b_2=y_j,$ which contradicts our assumption that they are distinct. The same analysis tells us that unless $\lambda_j+\lambda_j\mu_j=\mu_j,$ we also have $a_2=x_j.$ So for us to have at least $p+1$ elements in $\Theta_{1,2},$ we must have
$$\lambda_j + \lambda_j\mu_j=\mu_j \ \text{for} \ j=1,\dots, p-1.$$
Since $x_i\neq x_j$ and $y_i\neq y_j$ for any $i\neq j,$ we have $\lambda_i\neq \lambda_j$ and $\mu_i\neq \mu_j.$ So, the set $\Lambda\triangleq\{(\lambda_j,\mu_j):j=1,\dots, p-1\}$ has size $p-1$ and it is a subset of the solution set $\mathcal{S}\triangleq\{(\lambda,\mu)\in \F_p^2: \lambda+\lambda\mu=\mu, \lambda\neq 0,1,\mu\neq 0,-1\}.$ Thus, $|\mathcal{S}|\geq p-1.$
On the other hand, note that the relation $\lambda+\lambda\mu=\mu$ provides a one-to-one correspondence between $\lambda\in\F_p$ and $\mu\in \F_p.$ Based on the restriction that $\lambda\neq 0,1,$ there are at most $p-2$ choices for $\lambda.$ Thus, $|\mathcal{S}|\leq p-2$, contradicting the previous observation that $|\mathcal{S}|\geq p-1$.

\medskip

Now we are ready to prove Inequality \eqref{lowerbound}.
Using the Inequality \eqref{union}, we derive the bound for $|\delta S\cap (S+\gamma)|$ in different cases.
\begin{itemize}
\item Case $1:$ $(\delta,\gamma)\in (F_1)^2.$
\begin{itemize}
\item $\Theta_{1,1}.$ Since $\Theta_{1,1}=\delta G_1\cap (G_1+\gamma) \subseteq \delta G_1,$ we have $|\Theta_{1,1}|\leq |\delta G_1|=|G_1|=n/2.$
\item $\Theta_{1,2}.$ If there exist $a\in G_1$ and $b\in G_2$ such that $\delta a=b+\gamma$, then $b=\delta a - \gamma\in F_1\cap G_2=\emptyset$, a contradiction. So $|\Theta_{1,2}|=0.$
\item $\Theta_{2,1}.$ If there exist $a\in G_1$ and $b\in G_2$ such that $\delta b = a+\gamma$, then $b=\delta^{-1}(a+\gamma)\in F_1\cap G_2=\emptyset.$ So $|\Theta_{2,1}|=0.$
\item $\Theta_{2,2}.$ We claim that $|\Theta_{2,2}|\leq 1$. Suppose that there exist two distinct elements $(a_1,b_1)$ and $(a_2,b_2)$ in $G_2\times G_2$ such that $\delta a_i=b_i+\gamma, i=1,2.$
Hence $\delta=(b_1-b_2)/(a_1-a_2) \in F_1\cap (F_2\setminus\{0\})=\F_p\setminus \{0\}$. By the observation $F_1\cap G_2^*=\emptyset$ from (ii), we have $\gamma=\delta a_1-b_1\in F_1\cap (G_2^*\cup \{0\})=\{0\}$, that is, $\gamma=0$. Note that $\delta \neq 0$, hence both $a_1$ and $\delta a_1=b_1$ belong to $G_2$, which implies $\delta=1$ by observation (iii). We have $(\delta,\gamma)=(1,0)$, a contradiction. Thus $|\Theta_{2,2}|\leq 1$.
\end{itemize}
We obtain $|\delta S\cap (S+\gamma)|\leq n/2+1$ in this case.

\item Case $2:$ $\delta\in F_1$, $\gamma\not\in F_1.$

\begin{itemize}

\item $\Theta_{1,1}.$ By observation (i), we have $|\Theta_{1,1}| = 0.$
\item $\Theta_{1,2}.$ If there exist two distinct elements $(a_1,b_1)$ and $(a_2,b_2)$ in $G_1\times G_2$ such that $\delta a_i=b_i+\gamma$, $i=1,2$, then $\delta (a_1-a_2)=b_1-b_2 \in F_1\cap G_2 \subseteq F_1\cap G_2^*=\emptyset$, where the last equation follows from observation (ii), a contradiction. So, $|\Theta_{1,2}|\leq 1$.
\item $\Theta_{2,1}.$ By the same argument as in the case of $\Theta_{1,2},$ we have $|\Theta_{2,1}|\leq 1.$
\item $\Theta_{2,2}.$ As $\Theta_{2,2}=\delta G_2\cap (G_2+\gamma) \subseteq \delta G_2$, we have $|\Theta_{2,2}|\leq \delta |G_2|$. 

If $|\delta_{2,2}|\leq |G_2|-1$, then by combining with the analysis of $\Theta_{1,1}, \Theta_{1,2}$ and $\Theta_{2,1}$, we obtain $|\delta S\cap (S+\gamma)|\leq \sum_{i,j\in \{1,2\}} |\Theta_{i,j}|\leq |G_2|+1=n/2+1.$ Now assume that $|\Theta_{2,2}|= |G_2|=p^{t_1-2}\geq 2$. Let $(a_1,b_1)$ and $(a_2,b_2)$ be two distinct elements in $G_2\times G_2$ such that $\delta a_i=b_i+\gamma$, $i=1,2$. Hence, $\delta=(b_1-b_2)/(a_1-a_2) \in F_2$ and $\gamma=\delta a_1-b_1\in F_2$. Similar to the analysis of Case 1, we have $|\delta S\cap (S+\gamma)|\leq |G_2|+1=n/2+1.$
\end{itemize}

We obtain  $|\delta S\cap (S+\gamma)|\leq n/2+1$ in this case.  

\item Case $3:$ $\delta \not\in F_1$, $\gamma\in F_1.$ \\
In this case, we have $|\Theta_{1,1}|=0$ by observation (i).
For $\Theta_{1,2},\Theta_{2,1}$ and $\Theta_{2,2},$ we further divide this case into smaller cases depending on the relation between $\delta,\gamma$ and $F_2.$
\begin{itemize}
\item Case $3.1:$ $\delta\in F_2\setminus F_1,\gamma\in F_1\cap F_2.$ 
\begin{itemize}
\item $\Theta_{2,2}$. By similar analysis as in Case $1,$ we have $|\Theta_{2,2}|\leq |G_2|.$
\item $\Theta_{1,2}.$ If there exist $a\in G_1$ and $b\in G_2$ such that $\delta a=b+\gamma$, then $a=\delta^{-1}(b+\gamma)\in F_2\cap G_1=\emptyset$, a contradiction. Thus $|\Theta_{1,2}|=0$.
\item $\Theta_{2,1}.$ If there exist $a\in G_1$ and $b\in G_2$ such that $\delta b=a+\gamma$, then $a=\delta b-\gamma \in F_2\cap G_1=\emptyset$, a contradiction. Thus $|\Theta_{2,1}|=0.$
\end{itemize}
We obtain $|\delta S\cap (S+\gamma)|\leq |G_2|=n/2$ in this case. 

\item Case $3.2:$ $\delta\in F_2\setminus F_1$ and $\gamma\in F_1\setminus F_2.$
\begin{itemize}
\item $\Theta_{2,2}$. By observation (i), we have $|\Theta_{2,2}|=0.$ 
\item $\Theta_{1,2}$ and $\Theta_{2,1}$. Similar to the analysis of Case 2, we have $|\Theta_{2,1}|\leq 1$ and $|\Theta_{1,2}|\leq 1.$
\end{itemize}
We obtain $|\delta S\cap (S+\gamma)|\leq 2$ in this case. 

\item Case $3.3:$ $\delta\not\in F_1\cup F_2$ and $\gamma\in F_1\cap F_2.$
\begin{itemize}
\item $\Theta_{2,2}$. By observation (i), we have $\Theta_{2,2}=\emptyset.$ 
\item $\Theta_{1,2}$ and $\Theta_{2,1}.$ By Claim~\ref{lemmaneq}, we have $|\Theta_{1,2}|+|\Theta_{2,1}|\leq 2p.$
\end{itemize}
We obtain $|\delta S\cap (S+\gamma)|\leq 2p$ in this case.

\item Case $3.4:$ $\delta\not\in F_1\cup F_2$ and $\gamma\in F_1\setminus F_2$.
\begin{itemize}
\item $\Theta_{2,2}$. If $|\Theta_{2,2}|\geq 2$, then there exist distinct elements $(a_i,b_i) \in F_2\times F_2$ such that $\delta a_i=b_i+\gamma$, which implies $\delta=\frac{b_1-b_2}{a_1-a_2}\in F_2,$ a contradiction. Thus $|\Theta_{2,2}|\leq 1.$
\item $\Theta_{1,2}$ and $\Theta_{2,1}.$ By Claim~\ref{lemmaneq}, we  have $|\Theta_{1,2}|+|\Theta_{2,1}|\leq 2p.$
\end{itemize}

We obtain $|\delta S\cap (S+\gamma)|\leq 2p+1$ in this case.
\end{itemize}
In summary, we obtain $|\delta S\cap (S+\gamma)|\leq \max\{n/2,2p+1\}$ for Case 3.

\item Case $4:$ $\delta,\gamma\not\in F_1.$\\ 
By similar argument as in the analysis of $\Theta_{2,2}$ in Case $3.4,$ we have $|\Theta_{1,1}|\leq 1.$ For $\Theta_{1,2},\Theta_{2,1}$ and $\Theta_{2,2},$ we further divide this case into smaller cases depending on the relation between $\delta,\gamma$ and $F_2.$
\begin{itemize}
\item Case $4.1:$ $\delta,\gamma\in F_2\setminus F_1.$ By similar argument as in Case $1,$ we obtain $|\delta S\cap (S+\gamma)|\leq n/2+1.$
\item Case $4.2:$ $\delta\in F_2\setminus F_1$ and $\gamma\not\in F_1\cup F_2.$ By similar argument as in Case $3.2$, we obtain $|\delta S\cap (S+\gamma)|\leq 3.$
\item Case $4.3:$ $\delta\not\in F_1\cup F_2$ and $\gamma\in F_2\setminus F_1.$ By observation (i) and Claim~\ref{lemmaneq}, we have $|\Theta_{2,2}|=0$ and $|\Theta_{1,2}|+|\Theta_{2,1}|\leq 2p.$
Hence $|\delta S\cap (S+\gamma)|\leq 2p+1.$

\item Case $4.4: \delta,\gamma\not\in F_1\cup F_2.$ By similar argument as in the analysis of $\Theta_{1,1}$ in Case $4$, we have $|\Theta_{2,2}|\leq 1.$ Combining with Claim~\ref{lemmaneq}, we obtain $|\delta S\cap (S+\gamma)|\leq 2p+2.$
\end{itemize}
We obtain $|\delta S\cap (S+\gamma)|\leq \max\{n/2+1,2p+2,3\}=\max\{n/2+1,2p+2\}$ for Case 4.
\end{itemize}
Summarizing the results of the four cases, we obtain
$$|\delta S\cap (S+\gamma)|\leq \max\left\{\frac{n}{2}+1,2p+2,3\right\}=\max\left\{\frac{n}{2}+1,2p+2\right\}=\frac{n}{2}+1,$$
where in the last equation we use $n=2p^{t_1-2}$ and $t_1\geq 5$. Thus, the Inequality \eqref{lowerbound} is proved.

\medskip

Lastly, it remains to show that the code $\mC$ has length $n=2p^{t_1-2}=\exp(\sqrt{\log(q)})$. As $t_2=\Theta(t_1)$, we have $e=t_1t_2=\Theta(t_1^2)=\Theta((\log n)^2)$, which implies $q=\exp(\Theta((\log n)^2))$. Hence $n=\exp(\sqrt{\log q})$.

\end{proof}

\begin{rmk}
%We not only improve the construction in~{\cite[Corollary 1]{TS2007}}, which gives an insdel Reed-Solomon code over $\F_q$ of length $n=O(q)$ with insdel error correcting capability up to $O\left(\log n\right),$ but we also enhance the constructions  in {\cite[Theorem 3] {TS2007}} and {\cite[Theorem 4] {TS2007}}, which provide subcodes of Reed-Solomon codes over $\F_q$ of size $2q,$ length $n=O(q)$ with insdel error correcting capability up to $\frac{n}{3}.$ 
The insdel error-correcting capability of the Reed-Solomon code constructed in Lemma \ref{EC} is $\lfloor (n-3)/2 \rfloor$, {which outperforms the insdel error-correcting capabilities of Reed-Solomon codes and subcodes contructed in ~\cite{TS2007}}. The Reed-Solomon codes constructed in ~\cite[Corollary 1]{TS2007} have insdel error-correcting capability $O\left(\log n\right)$ and its subcodes ~\cite[Theorem 3, Theorem 4]{TS2007} have capability $n/3-1$.
The cost of this improvement comes in the {increment} of the field size compared to the code length. While our construction requires $q=\exp((\log n)^2),$ the previous constructions allow $q$ to be as small as $n.$
\end{rmk}

\medskip

Now we are ready to prove Theorem \ref{main1}. 
\mainone*
\begin{proof}
Let $s\in \Z^+$ such that $\epsilon>2^{-s}$. Write $e=t_1\cdots t_{2^s}$, where $t_1\geq 5$ and $t_1,\dots, t_{2^s}$ are consecutive primes. Let $p$ be a prime and put $q=p^e.$ To prove Theorem \ref{main1}, it suffices to construct a $2$-dimensional Reed-Solomon code over $\F_q$ with length $n=\exp\left((\log q)^{\frac{1}{2^s}}\right)$ and insdel distance $d$ satisfying 
\begin{equation} \label{insdellower}
d\geq 2\left(n\left(1-\frac{1}{2^s}\right)-(2^s-1)\left(\frac{p}{3}\cdot 2^s-\frac{2p}{3}+1\right)\right).
\end{equation}
Indeed, assuming the existence of the above code, we have $n=\exp\left((\log q)^{\frac{1}{2^s}}\right)=\exp((\log q)^\epsilon)$ and 
$$d\geq 2n\left(1-\frac{1}{2^s}\right)+O(2^{2s})\geq 2(1-\epsilon)n+O(1/\epsilon^2).$$
From now on, we focus on constructing $2$-dimensional Reed-Solomon codes over $\F_q$ with length $n=\exp\left((\log q)^{1/2^s}\right)$ and insdel distance $d$ satisfying Inequality \eqref{insdellower}.

\medskip

Denote $F=\F_q.$ For $i=1,\cdots, 2^s,$ we denote $F_i=\F_{p^{t_i}}$ and suppose that $F_i$ is the splitting field of $\beta_i$ over $\F_p$, that is, $F_i=\mathrm{span}_{\F_p}\{1,\cdots, \beta_i,\cdots,\beta_i^{t_i-1}\}.$ We further denote $G_i^\ast = \mathrm{span}_{\F_p}\{\beta_i,\cdots,\beta_i^{t_1-1}\}\setminus\{0\}$ and $G_i=\beta_i+\mathrm{span}_{\F_p}\{\beta_i^2,\cdots,\beta_i^{t_1-1}\}.$ For $1\leq i<j\leq 2^s,$ define $F_{[i,j]}=\F_{p^{\prod_{k=i}^j t_k}}$ and $G_{[i,j]}=\bigcup_{k=i}^j G_k.$ Lastly, we define $S=G_{[1,2^s]}.$ As before, $S$ represents the set of the $n$ values $\alpha_1,\cdots,\alpha_n$ used to construct the code $\mathcal{C}.$ In this case, we have $n=2^s|G_1|=2^sp^{t_1-2}.$ Note that
$$2\frac{n}{2^s}+2(2^{s}-1)\left(\frac{p}{3}\cdot 2^s-\frac{2p}{3}+1\right)=|G_1|+\sum_{i=0}^{s-1}\left(2^i(2^i-1)p+2^i\right).$$
By {Inequality} \eqref{insdellowerbound}, it suffices to prove the following
\begin{equation} \label{key}
\max_{\delta,\gamma\in F,\delta\neq 0,(\delta,\gamma)\neq(1,0)}|\delta S\cap (S+\gamma)|\leq |G_1|+\sum_{i=0}^{s-1}\left(\left(2^i(2^i-1)p\right)+2^i\right).
\end{equation}
We prove Inequality \eqref{key} by induction on $t.$  More precisely, we prove that for any $1\leq t\leq s$, $1\leq x\leq 2^s-2^t+1,$ and $\delta,\gamma\in F_{[x,x+2^t-1]}$ with $\delta\neq 0$, $(\delta,\gamma)\neq (1,0)$, we have 
\begin{equation} \label{induction}
|\delta G_{[x,x+2^t-1]} \cap (G_{[x,x+2^t-1]}+\gamma)|\leq |G_1|+\sum_{i=0}^{t-1}\left(\left(2^i(2^i-1)p\right)+2^i\right).
\end{equation}
Note that Inequality \eqref{key} follows from Inequality \eqref{induction} by taking $t=s$ and $x=1$.
For $t=1$, Inequality \eqref{induction} is proved by Lemma~\ref{EC}. Suppose that the claim is true for $t=l\leq s.$ This implies that for any $1\leq x\leq 2^s-2^l+1,$ and any $\delta,\gamma\in F_{[x,x+2^l-1]}$ with $\delta \neq 0, (\delta,\gamma)\neq (1,0)$, we have $|\delta G_{[x,x+2^l-1]}\cap(G_{[x,x+2^l -1]}+\gamma)|\leq |G_1|+\sum_{i=0}^{l-1}\left(\left(2^i(2^i-1)p\right)+2^i\right).$ 
Now we prove the claim for $t=l+1\leq s.$ Fix $x\in\{1,\cdots, 2^s-2^{l+1}+1\}.$ For simplicity of notation, put
$$H_1 = G_{[x,x+2^{l}-1]}, \ H_2 = G_{[x+2^{l},x+2^{l+1}-1]}, \ H = H_1\cup H_2, $$
$$E_1 = F_{[x,x+2^{l}-1]}, \ E_2 = F_{[x+2^{l},x+2^{l+1}-1]}, \ E= F_{[x,x+2^{l+1}-1]}.$$
It is clear $H_1\cap H_2=H_1\cap E_2=H_2\cap E_1=\emptyset$ and $E_1\cap E_2=\F_p.$ Our aim is to prove that for any $\delta,\gamma\in E,\ \delta\neq 0,\ (\delta,\gamma)\neq (1,0)$, we have 
$$|\delta H\cap (H+\gamma)|\leq |G_1|+\sum_{i=0}^{l}\left(\left(2^i(2^i-1)p\right)+2^i\right).$$ 
For any $i,j\in \{1,2\},$ denote $\Theta_{i,j}= \delta H_i \cap (H_j+\gamma).$ Hence $|\delta H\cap (H+\gamma)|\leq \sum_{i=1}^{2}\sum_{j=1}^2 |\Theta_{i,j}|.$ As before, we consider into $4$ cases depending on the values of $\delta$ and $\gamma$ to investigate the numbers $|\Theta_{i,j}|$.
\begin{itemize}
\item Case 1: $\delta,\gamma\in E_1.$
\begin{itemize}
\item $\Theta_{1,1}$. By the induction hypothesis, we have $|\Theta_{1,1}|\leq |G_1|+\sum_{i=0}^{l-1}\left(\left(2^i(2^i-1)p\right)+2^i\right).$
\item $\Theta_{1,2}.$ If there exist $a\in H_1$ and $b\in H_2$ such that $\delta a=b+\gamma$, then $b= \delta a-\gamma\in E_1\cap H_2=\emptyset,$ a contradiction. Thus $|\Theta_{1,2}|=0.$
\item $\Theta_{2,1}.$ Using similar argument as $\Theta_{1,2},$ we have $|\Theta_{2,1}|=0.$
\item $\Theta_{2,2}.$ Note that $|\Theta_{2,2}|\leq \sum_{i=x+2^l}^{x+2^{l+1}-1}\sum_{j=x+2^l}^{x+2^{l+1}-1} |\delta G_i\cap (G_j+\gamma)|.$ By Lemma~\ref{lemmaneq}, we have  $|\delta G_i\cap (G_j+\gamma)|\leq p$ for any $i\neq j$. A similar argument as in the analysis of $\Theta_{2,2}$ in Case $1$ of Lemma~\ref{EC} gives $|\delta G_i\cap (G_i+\gamma)|\leq 1.$ Hence $|\Theta_{2,2}|\leq 2^l(2^l-1)p + 2^l.$
\end{itemize} 
We obtain $|\delta H \cap (H+\gamma)|\leq |G_1|+\sum_{i=0}^{l-1}\left(\left(2^i(2^i-1)p\right)+2^i\right)+2^l(2^l-1)p + 2^l=|G_1|+\sum_{i=0}^l\left(2^i(2^i-1)p+2^i\right)$ in this case.
\end{itemize}

\begin{itemize}

\item Case $2: \delta\in E_1,\gamma\not\in E_1.$
\begin{itemize}
\item $\Theta_{1,1}:$ If there exist $a,b\in H_1$ such that $\delta a=b+\gamma,$ then $\gamma=\delta a-b\in E_1,$ contradicting the assumption that $\gamma\not\in E_1.$ So $|\Theta_{1,1}|=0.$
\item $\Theta_{1,2}.$ Note that the argument in Claim~\ref{lemmaneq} still holds if we replace the sets $G_1$ and $G_2$ by $H_1$ and $H_2$, respectively. Hence $|\Theta_{1,2}|\leq p.$
\item $\Theta_{2,1}.$ Using the same argument as before, we have $|\Theta_{2,1}|\leq p.$
\item $\Theta_{2,2}.$  If $|\Theta_{2,2}|\leq 1$, then we have $\delta H\cap (H+\gamma)|\leq 2p+1$.\\ 
Suppose that $|\Theta_{2,2}|\geq 2$. There exist two distinct elements $(a_i,b_i)\in H_2\times H_2$, $i=1,2,$ such that $\delta a_i=b_i+\gamma$, which implies $\delta=(b_1-b_2)/(a_1-a_2)\in E_2$ and $\gamma=\delta a_1-b_1\in E_2$. By Case 1, we have $\delta H\cap (H+\gamma)|\leq |G_1|+\sum_{i=0}^l \left(2^i(2^i-1)p+2^i\right)$.

\end{itemize}
We have $|\delta H\cap (H+\gamma)|\leq |G_1|+\sum_{i=0}^{l}\left(2^i(2^i-1)p+2^i\right)$ in this case.
\end{itemize}

\begin{itemize}

\item Case $3: \delta\not\in E_1$, $\gamma\in E_1.$\\ 
If there exists $a,b\in H_1$ such that $\delta a=b+\gamma,$ we have $\delta = a^{-1}(b+\gamma)\in E_1,$ a contradiction. So $|\Theta_{1,1}|=0.$
\begin{itemize}
\item Case $3.1: \delta\in E_2\setminus E_1,\gamma\in E_1\cap E_2.$ By similar argument asin  Case $1,$ we have $|\Theta_{2,2}|\leq |G_1|+\sum_{i=0}^{l-1}\left(2^i(2^i-1)p+2^i\right)$ and $|\Theta_{1,2}|+|\Theta_{2,1}|=0.$ Hence $|\delta H\cap (H+\gamma)|\leq |G_1|+\sum_{i=0}^{l-1}\left(2^i(2^i-1)p+2^i\right).$
\item Case $3.2: \delta\in E_2\setminus E_1,\gamma\in E_1\setminus E_2.$ By similar argument as in Case $2,$ we have $\Theta_{2,2}=\emptyset, |\Theta_{1,2}|\leq p$ and $|\Theta_{2,1}|\leq p$. Hence $|\delta H\cap (H+\gamma)|\leq 2p.$
\item Case $3.3: \delta \not\in E_1\cup E_2,\gamma\in E_1\cap E_2.$ By similar argument as in the analysis of $\Theta_{1,1}$ in Case $3,$ we have $|\Theta_{2,2}|=0$. By Claim~\ref{lemmaneq}, we have $|\Theta_{1,2}|\leq p$ and $|\Theta_{2,1}|\leq p.$ Hence $|\delta H\cap (H+\gamma)|\leq 2p.$
\item Case $3.4: \delta\not\in E_1\cup E_2,\gamma\in E_1\setminus E_2.$ By similar argument as in the analysis of $\Theta_{2,2}$ in Case $1,$ we have $|\Theta_{2,2}|\leq 2^l(2^l-1)p+2^l.$ By Claim~\ref{lemmaneq}, we have $|\Theta_{1,2}|\leq p$ and $|\Theta_{2,1}|\leq p.$ Hence, $|\delta H \cap (H+\gamma)|\leq 2p+2^l(2^l-1)p+2^l.$
\end{itemize}

We obtain in this case, 
\begin{eqnarray*}
|\delta H \cap (H+\gamma)| &\leq& \max\Big\{ |G_1|+\sum_{i=0}^{l-1}\left(2^i(2^i-1)p+2^i\right), 2p+2^l(2^l-1)p+2^l\Big\} \\
& <& |G_1|+\sum_{i=0}^{l}\left(2^i(2^i-1)p+2^i\right).
\end{eqnarray*}
\end{itemize}

\begin{itemize}
\item Case $4:$ $\delta,\gamma\not\in E_1.$\\ 
By similar argument as in the analysis of $\Theta_{2,2}$ in Case $3.4,$ we have $|\Theta_{1,1}|\leq 2^l(2^l-1)p+2^l.$ For $\Theta_{1,2}, \Theta_{2,1}$ and $\Theta_{1,1}$, we consider the following sub-cases. 
\begin{itemize}
\item Case $4.1:$ $\delta,\gamma\in E_2\setminus E_1.$ By similar argument as in Case $1,$ we have $|\delta H\cap (H+\gamma)|\leq |G_1|+\sum_{i=0}^{l}\left(2^i(2^i-1)p+2^i\right).$
\item Case $4.2:$ $\delta\in E_2\setminus E_1$, $\gamma\not\in E_1\cup E_2.$ This case is similar as Case $3.2$, so $|\delta H\cap (H+\gamma)|\leq 2p+2^l(2^l-1)p+2^l.$
\item Case $4.3:$ $\delta\not\in E_1\cup E_2$, $\gamma\in E_2\setminus E_1.$ In this case, we have $\Theta_{2,2}=\emptyset$ and $|\Theta_{1,2}|\leq p$ and $|\Theta_{2,1}|\leq p.$
Hence $|\delta H\cap (H+\gamma)|\leq 2p+2^l(2^l-1)p+2^l.$
\item Case $4.4: \delta,\gamma\not\in E_1\cup E_2.$ In this case, we have $|\Theta_{2,2}|\leq 2^l(2^l-1)p+2^l$, $|\Theta_{1,2}|\leq p$ and $|\Theta_{2,1}|\leq p$.
Hence $|\delta H\cap (H+\gamma)|\leq 2p+2\cdot(2^l(2^l-1)p+2^l).$
\end{itemize}
We obtain $|\delta H\cap (H+\gamma)|\leq |G_1|+\sum_{i=0}^{l}\left(2^i(2^i-1)p+2^i\right)$ in this case.

\end{itemize}
Summarizing the results of the four cases, we obtain
$$|\delta H \cap (H+\gamma)|\leq |G_1|+\sum_{i=0}^l((2^i(2^i-1)p)+2^i),$$ 
proving Inequality \eqref{induction}. To finish the proof, we need to show that $n=\exp\left((\log q)^{1/2^s}\right)$. Note that $n=2^s p^{t_1-2}$, $q=p^e$  and $e=t_1\cdots t_{2^s}$ in which $5\leq t_1<\cdots<t_{2^s}$. Assuming $t_i=O(t_1)=O(\log n)$ for all $i,$ we have $q=\exp\left((\log n)^{2^s}\right).$ We obtain $n=\exp\left((\log q)^{{1}/{2^s}}\right)$.
\end{proof}

\bigskip

\section{Reed-Solomon Codes with Insdel Error-Correcting Capability up to its Length} \label{sec:constr2}

%{\color{red}First, we state the main result of this section, which is Theorem \ref{main2}.

%\begin{theorem} \label{main22}
%Let $q=p^s$ be a power of an odd prime $p$. If $q-1$ contains a nontrivial divisor $f$ such that $f<\log_{\sqrt{14}}(p)$ and $f$ is proportional to $\log_{\sqrt{14}}(q)$,
%then there exists a $2$-dimensional Reed-Solomon code $\mC$ over $\F_q$ with length $n$ and insdel distance $d$ satisfying
%$$n=\Theta\left( \sqrt{\log q} \right) \ \ \text{and}\ \ d\in \{2n-6,2n-4\}.$$
%\end{theorem}}

The aim of this section is to prove Theorem~\ref{main2}, which claims the existence of an $[n,2,2n-4]_q$-insdel Reed-Solomon code $\mC$ whose minimum insdel distance $d=2n-4$ meets the bound provided by Theorem \ref{improvedbound}. A combination of Theorem \ref{improvedbound} and  {Inequality} (\ref{insdellowerbound}) implies that $\mC$  exists if there {is} a subset $S$ of $\F_q$ such that $|S|=n$ and $|\delta S\cap (S+\gamma)|\leq 2$ for any $(\delta,\gamma)\in \F_q\times \F_q\setminus\{(1,0)\}$. In general, there are many choices for a subset $S$ of $\F_q$ with the property $|\delta S\cap (S+\gamma)|\leq 2$ for any $(\delta,\gamma)\in \F_q\times \F_q\setminus\{(1,0)\}$. Among these choices, we would like to have $|S|$ to be as large as possible so that the length $n$ of the code $\mC$ is large.
The construction of $S$ can be classified into the following two steps.

\begin{itemize}
\item[1.] $S$ is chosen as a subset of $S'$, the cyclic subgroup of $\F_q^*$ of order $f$, such that 
\begin{equation}\label{assumption}
|\delta S\cap S|\leq 2 \ \text{for any} \ \delta\in \F_q\setminus\{1\}.
\end{equation}
Combining with Lemma \ref{cyclo_refined}, we can prove that $|\delta S\cap (S+\gamma)|\leq 2$ for any $(\delta,\gamma)\in \F_q\times \F_q\setminus \{(1,0)\}$. 
\item[2.] %Under the assumption {\color{magenta}in step {\color{red}delete step} $1$} {\color{blue} that $|\delta S\cap S|\leq 2$ for any $\delta \in \F_q\setminus\{1\}$}, {\color{magenta} we choose the set $S$ to have size $\Theta(\sqrt{\log q})$, which is asymptotically best possible.} {\color{green} This sentence has the implication that we know beforehand that the best size for $S$ is $\Theta(\sqrt{\log q})$ and we try to achieve this. It would cause questions like ``how do you know that it must be $\Theta(\sqrt{\log q})?$ How to prove this'', which is not really the direction that we are going from. Instead, this size is just the implication of our effort to find the best size of $S,$ not the assumption. In fact I believe ``$|S|=\Theta(\sqrt{\log q})$'' does not need to be mentioned here since it is not a detail that is required for the reader to understand the steps. However, if you want to add it, based on my understanding, the following sentences may be more accurate.}

%{\color{blue} Under the assumption {\color{red}\eqref{assumption},} we choose $S$ to have the largest possible size. It is shown in Theorem~\ref{main2} that the largest asymptotic size possible is $\Theta(\sqrt{\log q}).$}

Under the assumption~(\ref{assumption}), we choose $S$ to have size largest possible. It is shown in Theorem~\ref{main2} that the asymptotically largest possible size of $S$ is $\Theta(\sqrt{\log q})$.
\end{itemize}

\medskip

The following lemma indicates the first step. The choice of the subset $S$ with the property $|\delta S\cap (S+\gamma)|\leq 2$ for any $(\delta,\gamma)\in \F_q\times \F_q\setminus\{(1,0)\}$ reduces to the choice of a subset $T$ of $\Z_f$ such that $|T|=|S|$ and any nonzero element of $\mathbb{Z}_f$ appears at most twice in $T-T$, where $T-T=\{x-y: x,y\in T\}$ is considered as a multi-set.

\begin{lemma} \label{constr_main}
Let $p$ be a prime. Write $q=p^s=ef+1$ for some $s\in \mathbb{Z}^+,$ and nontrivial divisors $e$ and $f$ of $q-1.$ By ${\rm{ord}}_f(p)$ we denote the smallest positive integer $j$ such that $p^j \equiv 1\pmod{f}$. If 

\begin{itemize}
\item[(i)] $f<\log_{\sqrt{14}}\left(p^{{\rm{ord}}_f(p)}\right)$, and
\item[(ii)] there exists a subset $T$ of $\mathbb{Z}_f$ of size $n$ such that any nonzero element of $\mathbb{Z}_f$ appears at most twice in $T-T$, where $T-T=\{x-y: x,y\in T\}$ is considered as a multi-set,
\end{itemize}
then there exists a Reed-Solomon code $\mC$ over $\F_q$ of dimension $2$, length $n$ and insdel distance $d$ satisfying
$$d\geq 2n-4.$$
\end{lemma}

\begin{proof}
We use Lemma \ref{cyclo_refined} for the proof of this lemma. Recall that $g$ be a primitive element of $\F_q$. For each $a\in \Z$, we define $S_a=\{g^a,g^{a+e},\dots,g^{a+(f-1)e}\}$. Set
$$S'=S_0=\{1,g^e,\cdots, g^{(f-1)e}\}.$$
By Lemma \ref{cyclo_refined}, we have $|(1+S_a)\cap S_b|\leq 2$ for any $a,b\in \Z$.
First, we claim that 
\begin{equation} \label{intersect}
|\delta S'\cap (S'+\gamma)|\leq 2 \ \text{whenever} \ \gamma\neq 0.
\end{equation}
If $\delta=0$, then $|\delta S'\cap (S'+\gamma)|=|\{0\}\cap (S'+\gamma)| \leq 1$. Assume $\delta\neq 0$. We have $|\delta S'\cap (S'+\gamma)|=|\delta\gamma^{-1}S'\cap (\gamma^{-1}S'+1)|$. As $\delta\gamma^{-1}S'=S_a$ and $\gamma^{-1}S'=S_b$ for some $a,b\in \Z$, we obtain $|\delta S'\cap (S'+\gamma)|=|S_a\cap (1+S_b)|\leq 2$. The claim is proved.

\medskip

Recall that we need to choose $S\subset \F_q$ such that $|\delta S\cap (S+\gamma)|\leq 2$ for any $(\delta,\gamma)\in \F_q\times \F_q\setminus \{(1,0)\}$. First, we choose $S$ as a subset of $S'$. By (\ref{intersect}), we have
\begin{equation} \label{intersect2}
|\delta S\cap (S+\gamma)|\leq 2 \ \text{for any} \ \gamma\in \F_q\setminus\{0\}.
\end{equation}
It remains to choose $S$ in a way such that $|\delta S\cap S|\leq 2$ for any $\delta \in \F_q\setminus\{1\}$. In the following, we utilize the set $T$ given in the assumption {$(ii)$} to choose $S$.
Let
$$T=\{a_1,\dots,a_n\}$$
be the subset of $\Z_f$ such that any nonzero element of $\Z_f$ appears at most twice in the multi-subset $T-T$ of $\Z_f$. Define the subset $S$ of $S'$ as follows
$$S=\{g^{a_1e},\dots,g^{a_ne}\}.$$
We claim that $|\delta S\cap S|\leq 2$ for any $\delta\in \F_q\setminus\{1\}$. Assume that $|\delta S\cap S|\geq 3$ for some $\delta\in \F_q\setminus\{1\}$. There exist three distinct elements $(a_{i_t},a_{j_t}), t=1,2,3,$ in $T\times T$ such that $\delta g^{a_{i_t}e}=g^{a_{j_t}e}$ for any $t=1,2,3$. Note that $\delta \neq 1$, so $a_{i_t}\neq a_{j_t}$ for any $t$. We have
$\delta^{-1}=g^{(a_{i_1}-a_{j_1})e}=g^{(a_{i_2}-a_{j_2})e}=g^{(a_{i_3}-a_{j_3})e}.$
Hence
$$a_{i_1}-a_{j_1}= a_{i_2}-a_{j_2}=a_{i_3}-a_{j_3} \pmod{f},$$ where the common residue (modulo $f$) is nonzero, contradicting the choice of the set $T$. Therefore, we have $|\delta S\cap S|\leq 2$ for any $\delta\in \F_q\setminus\{1\}$. Combining with (\ref{intersect2}), we obtain
$$|\delta S\cap (S+\gamma)|\leq 2 \ \text{for any} \ (\delta,\gamma)\in \F_q\times \F_q\setminus \{(1,0)\},$$
which concludes our proof.
\end{proof}

\medskip

Lemma~\ref{constr_main} provides a $2$-dimensional Reed-Solomon code $\mC$ of length $n$ given the existence of a subset $T$ of $\Z_f$ such that $|T|=n$ and any nonzero element of $\Z_f$ appears at most twice in the multi-set $T-T$. In general, there are numerous choices of the set $T$ with this property. However, for the length $n$ of the code $\mC$ to be large, we would like to have the choice such that $|T|$ is as large as possible. A simple counting argument shows that $|T|\leq \lceil \sqrt{2f}\rceil$. In fact, we will show that we can choose $T$ with size $|T|\geq \lceil \sqrt{f}/2 \rceil$. This would result in the choice of the set $S$ with $|S|=|T|=\Theta(\sqrt{f})$. We are ready for proof of Theorem \ref{main2}.

\maintwo*

\begin{proof}
First, we claim that the subset $T$ in the assumption {$(ii)$} of Lemma \ref{constr_main} exists and $T$ has size satisfying
\begin{equation} \label{sizeofT}
 \lceil \sqrt{f}/{2} \rceil  \leq |T| \leq \lceil \sqrt{2f} \rceil.
\end{equation} 
The upper bound in \eqref{sizeofT} can be simply obtained as follows.
Since each nonzero element of $\Z_f$ appears at most three times in the multi-set $T-T=\{x-y: x,y\in T\}$, we have
\begin{equation} \label{inequal2}
|T|(|T|-1)\leq 2(f-1).
\end{equation}
If $\sqrt{2f}=m$ is an integer, then {Inequality} \eqref{inequal2} implies $|T|\leq m$. If $\sqrt{2f}$ is not an integer, then {Inequality}~\eqref{inequal2} implies $|T|<\sqrt{2f}+1$. So $|T|\leq \lceil \sqrt{2f} \rceil$ in any case.
It remains to show the existence of the subset $T$ of $\Z_f$ such that any nonzero element of $\Z_f$ appears at most twice in $T-T$ and
\begin{equation} \label{lowerbound2}
|T|\geq \lceil \sqrt{f}/2 \rceil.
\end{equation}
First, we consider the case $f\leq 8$. For ${2}\leq f \leq 4$, the set $T=\{0\}$ has size $|T|=1=\lceil \sqrt{f}/2\rceil$ and $T-T=\{0\}$ does not contain any nonzero element of $\Z_f$. For $5\leq f\leq 8$, the set $T=\{0,1\}$ has size $|T|=2=\lceil \sqrt{f}/2 \rceil$ and $T-T=\{0,1,-1\}$ contains any nonzero element of $\Z_f$ at most once. 
From now on, we assume $f\geq 9$. We claim that there exists a prime $r$ such that
\begin{equation} \label{bertrand}
\frac{\sqrt{f}}{2}\leq r\leq \sqrt{f}-1.
\end{equation}
If $9\leq f\leq 16$, we choose $r=2$. Assume that $f\geq 17$. Put $m=\lceil \sqrt{f}/2\rceil \geq 3$. By Bertrand's postulate \cite{ber}, there exists a prime $r$ such that $m\leq r\leq 2m-3$, which implies \eqref{bertrand}. 

\medskip

Next, we use Lemma \ref{singer} to contruct the set $T$. Applying Lemma~\ref{singer} for $d=3,$ we obtain a subset $D$ of a multiplicative cyclic group $G$ with the properties $|G|=r^2+r+1$, $|D|=r+1$ and every non-identity element of $G$ appears exactly once in the multi-set $DD^{(-1)}=\{xy^{-1}: x,y\in D\}$. 
Let 
$$\pi:G\rightarrow \Z_{r^2+r+1}$$
be an isomorphism between cyclic groups $G$ and $\Z_{r^2+r+1}$. Hence the subset $T=\pi(D)$ of $\Z_{r^2+r+1}$ has the properties that $|T|=r+1$ and every nonzero element of $\Z_{r^2+r+1}$ appears exactly once in the multi-subset $T-T$ of $\Z_{r^2+r+1}$. Note that $r^2+r+1<f$ by \eqref{bertrand}. We consider $T$ as a subset of $\Z_f$, that is, if $T=\{a_0,\dots,a_r\}$ is a subset of $\Z_{r^2+r+1}$, then we take the exact set $T=\{a_0,\dots, a_r\}$ as a subset of $\Z_f$. The only difference is that in $\Z_f$, the operation between elements of $T$ is calculated modulo $f$, so the property that any non-zero element in $\Z_f$ appears exactly once in the multi-set $T-T$ does not necessarily hold any more. By \eqref{bertrand}, it is clear that 
$$|T|=r+1>\sqrt{f}/2,$$
proving the lower bound \eqref{lowerbound2} on $|T|$.
It remains to show that any non-zero element of $\Z_f$ appears at most twice in the multi-subset $T-T$ of $\Z_f$.
Write $T=\{a_0,\dots,a_r\}$ with $0\leq a_0<\cdots<a_r\leq r(r+1)$ and write 
$$X=\{a_i-a_j: i>j\},$$
where $X$ is considered as a normal subset (not a multi-set) of $\Z_f$.
As $X\cup (-X)=\{x-y: x\neq y,\ x,y\in T\}$ has size at most $r(r+1)$ and covers all nonzero residues modulo $r^2+r+1$ (by the property of $T$), we have $|X|=r(r+1)/2$. Write $X=\{x_1,\dots,x_{r(r+1)/2}\}$.
Note that $1\leq x_i\leq r(r+1)$ for all $i$.
The set $(T-T)\setminus \{0\}=\{x-y: x\neq y, x,y\in T\}$, considered as a multi-subset of $\Z_f$, is
$$(T-T)\setminus\{0\}=\{x_1,\dots, x_{r(r+1)/2}, f-x_1,\dots,f-x_{r(r+1)/2}\}.$$
The only possible repetitions in the above set come from $x_i=f-x_j$ for some $i,j\in \{1,\dots, r(r+1)/2\}$. Therefore, any nonzero element of $\Z_f$ appears at most twice in the multi-subset $T-T$ of $\Z_f$. We finish the proof on the claim about the existence of $T$ with size satisfying (\ref{sizeofT}).

\medskip

Lastly, by (\ref{sizeofT}) and Lemma \ref{constr_main}, there exists a $2$-dimensional Reed-Solomon code $\mC$ over $\F_q$ with length $n=|T|=\Theta(\sqrt{f})$ and insdel distance $d\geq 2n-4$. Note that $q=p^s\equiv 1\pmod{f}$, so $s\equiv 0\pmod{{\rm{ord}}_f(p)}$. This implies $q\geq p^{{\rm{ord}}_f(p)}>\left(\sqrt{14}\right)^f>n^2$. By Theorem \ref{improvedbound}, we have $d\leq 2n-4$. Hence $d=2n-4$, proving the claim on the minimum insdel distance of $\mC$. On the other hand, as $f=\Theta(\log q)$ by the assumption $(ii)$ of Theorem \ref{main2}, we obtain $n=\Theta(\sqrt{f})=\Theta(\sqrt{\log q})$, proving the remaining claim on the length $n$ of the code $\mC$. 
\end{proof}

\bigskip

\section{Conclusion}\label{sec:conc}
In this paper, {we proved an improved bound on the minimum insdel distance of an $[n,k,d]_q$-insdel Reed-Solomon codes with $q\geq n^2$. In this case, we obtain $d\leq 2n-2k$, which implies that these codes can never achieve the equality $d=2n-2k+2$ in the Singleton bound (\ref{singleton-insdel})}.
Furthermore, we constructed two explicit families of $2$-dimensional insdel Reed-Solomon codes whose insdel error-correcting capabilities asymptotically {reaching} $n-3$, the capability provided by our {upper} bound on the minimum insdel distance of these codes. Our constructions improve the previously known construction of $2$-dimensional Reed-Solomon codes whose insdel error-correcting capability is only logarithmic in $n$.
To end our paper, we would like to propose several research directions in the study of Reed-Solomon codes under insdel metric.

\begin{enumerate}
\item Both our constructions in Theorems~\ref{main1} and~\ref{main2} focus on codes of dimension $2$. Although arguments presented here may be applied to a higher dimension, it leads to a much more complex analysis. An open question is to construct insdel Reed-Solomon codes with dimension $k\geq 3$ and insdel {error-correcting} capabilities better than $\log_{k+1}n$, the previously known {error-}correcting capability provided in \cite{TS2007}.

\item The deletion correcting algorithm for Reed-Solomon codes has been discussed in~\cite{SW2003} and\cite{2004}. Despite the equivalence between $t$-deletion error-correcting capability and $t$-insdel error-correcting capability, there has not been any decoding algorithm against $t$-insdel error-correcting capability. Moreover, another open question is the construction of a list decoding algorithm of Reed-Solomon codes under insdel metric to investigate the limit of list decodability of insdel Reed-Solomon codes.

\item It is well-known that under Hamming metric, the dual of a Reed-Solomon code is also a Reed-Solomon code. To the best of our knowledge, there has not been any study on the relation between the insdel distance and its dual insdel distance. An investigation of this relation may reveal more insight into insdel Reed-Solomon codes.
\end{enumerate}

\end{document}